\documentclass[conference]{IEEEtran}
\IEEEoverridecommandlockouts
\usepackage{cite}
\usepackage{amsmath,amssymb,amsfonts}
\usepackage{algorithmic}
\usepackage{graphicx}
\usepackage{textcomp}
\usepackage{xcolor}
\usepackage{amsthm}
\newtheorem{lemma}{Lemma}
\newtheorem{corollary}{Corollary}
\newtheorem{proposition}{Proposition}
\newtheorem{theorem}{Theorem}
\newtheorem{remark}{Remark}
\usepackage{comment}
\usepackage{balance}
\def\BibTeX{{\rm B\kern-.05em{\sc i\kern-.025em b}\kern-.08em
    T\kern-.1667em\lower.7ex\hbox{E}\kern-.125emX}}

\begin{document}

\title{Transmit or Idle: Efficient AoI Optimal Transmission Policy for Gossiping Receivers }

\author{Irtiza Hasan$\qquad$Ahmed Arafa\\Department of Electrical and Computer Engineering\\ University of North Carolina at Charlotte, NC 28223\\
\emph{ihasan@charlotte.edu}$\qquad$\emph{aarafa@charlotte.edu}
\thanks{This work was supported by the U.S. National Science Foundation under Grants CNS 21-14537 and ECCS 21-46099.}}

\maketitle


\begin{abstract}
We study the optimal transmission and scheduling policy for a transmitter (source) communicating with two \emph{gossiping} receivers aiming at tracking the source's status over time using the age of information (AoI) metric. Gossiping enables local information exchange in a decentralized manner without relying solely on the transmitter's direct communication, which we assume incurs a transmission cost. On the other hand, gossiping may be communicating stale information, necessitating the transmitter's intervention. With communication links having specific success probabilities, we formulate an average-cost Markov Decision Process (MDP) to jointly minimize the sum AoI and transmission cost for such a system in a time-slotted setting. We employ the Relative Value Iteration (RVI) algorithm to evaluate the optimal policy for the transmitter and then prove several structural properties showing that it has an \emph{age-difference threshold structure with minimum age activation} in the case where gossiping is relatively more reliable. Specifically, direct transmission is optimal only if the minimum AoI of the receivers is large enough \emph{and} their age difference is below a certain threshold. Otherwise, the transmitter idles to effectively take advantage of gossiping and reduce direct transmission costs. Numerical evaluations demonstrate the significance of our optimal policy compared to multiple baselines. Our result is a first step towards characterizing optimal freshness and transmission cost trade-offs in gossiping networks.
\end{abstract}


\section{Introduction}

Information freshness in status updating systems such as sensor networks, autonomous vehicles, and the Internet-of-Things (IoT) is critical for successful control, estimation and decision-making. Age of Information (AoI), defined as the time elapsed since the latest delivered piece of information has been generated at its source, has emerged as an important metric for tracking timeliness of information delivery in such systems \cite{AoISurvey}. The aforementioned and similar other networked-systems can be modeled as gossip networks where receivers share their information among themselves---a process we term as \emph{gossiping}. This gossiping behavior among nodes opens up an interesting choice for information dissemination at transmitters: with the knowledge that receivers may gossip, transmitters can choose to rely on them to partly do their own job of propagating information to other receivers. Instead of costly transmissions from a transmitter to each receiver, such transmitter can idle while receivers gossip. On the other hand, gossiping might be occurring using relatively stale data since transmitters may have fresher information that they have not yet shared. This calls for transmitters to operate within a policy to save transmission costs with a trade-off of information freshness in the network. Our goal in this paper is to introduce and characterize this trade-off in an optimal manner.

Information freshness studies in gossiping have analyzed how topology, cluster formation, and mobility patterns affect the freshness in gossip networks  \cite{AoG,VAoIClustered,RCGossip,VAoIContact,AoGSurvey}. 
One recent work \cite{VAoIPolicy} integrates freshness-aware control policies in such a system by formulating a Markov Decision Process (MDP) for an energy-harvesting sensor and a cache-enabled aggregator, showing that the optimal policy has a \emph{threshold} structure. Characterizing the \emph{structure} of optimal policies for freshness in status updating systems has been studied extensively in different settings. For example, in RF-powered IoT networks, the works in \cite{AbdElmagidTWC2019, AbdElmagidJoint2020} establish the monotonicity of their modeled MDP's value function and show that the optimal sampling and transmission strategy is threshold-based with respect to state variables. In computation offloading, reference~\cite{JeffModiano2025} demonstrates that the delay-optimal policy also has a threshold structure. In a transmitter-receiver pair scenario, using the Age of Incorrect Information (AoII) metric, the optimal transmission policy is shown to have a \emph{randomized threshold} structure in \cite{MaatoukAoII2022}. Knowing the optimal policy's structure is significantly useful, e.g., for implementing Reinforcement Learning (RL)-based methods as has been shown in \cite{ThresholdNW1, ThresholdNW2, StrcutureRL}. 

\begin{figure}[t]
    \centering
    \includegraphics[width=0.8\linewidth]{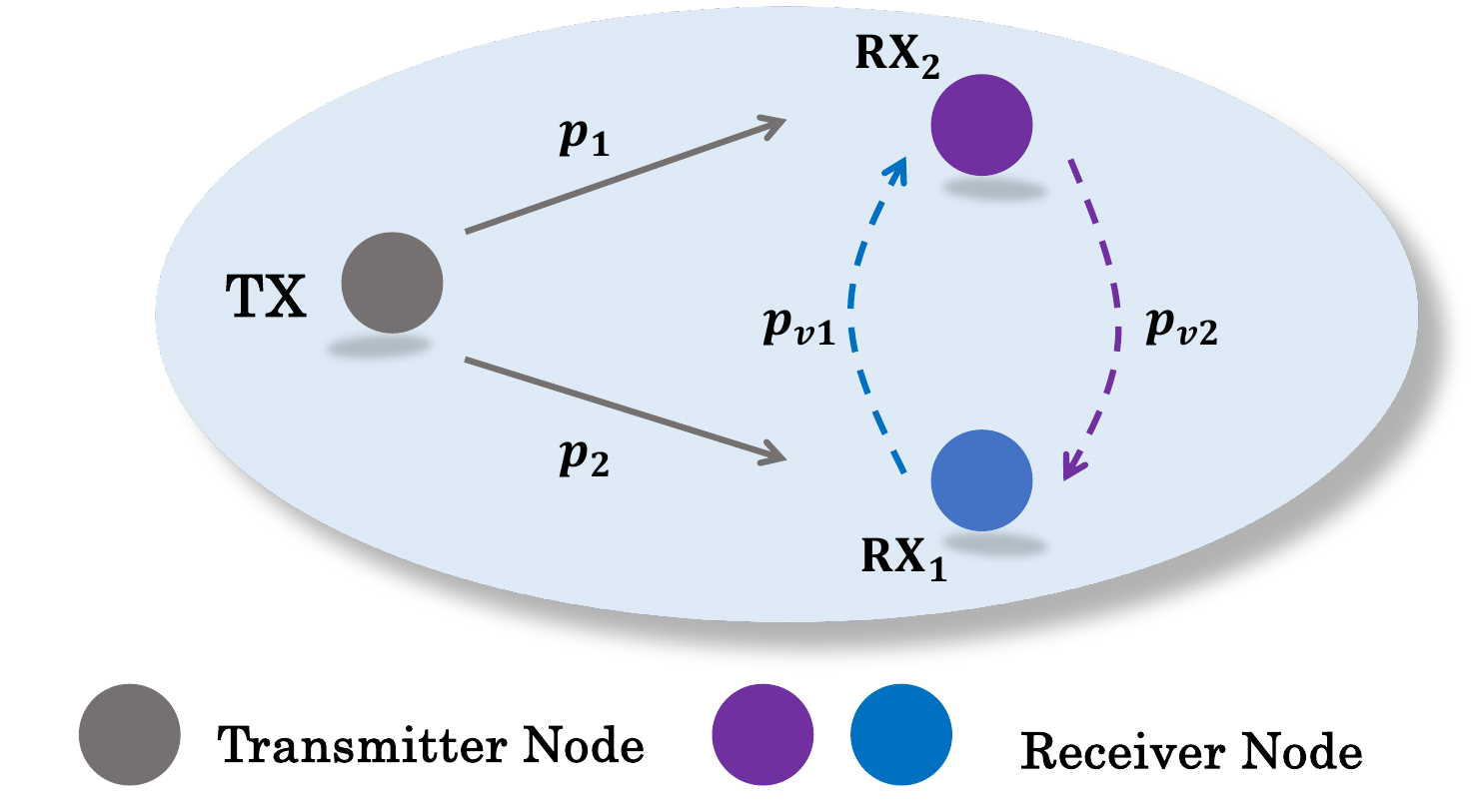} 
    \caption{Transmitter updating two receivers via direct links (solid lines) and receivers gossiping (dashed lines). Symbols above the lines indicate channels' successful communication probabilities.}
    \label{fig:GR System Model}
\end{figure}

\begin{figure*}[t]
    \centering
    \includegraphics[width=0.8\linewidth]{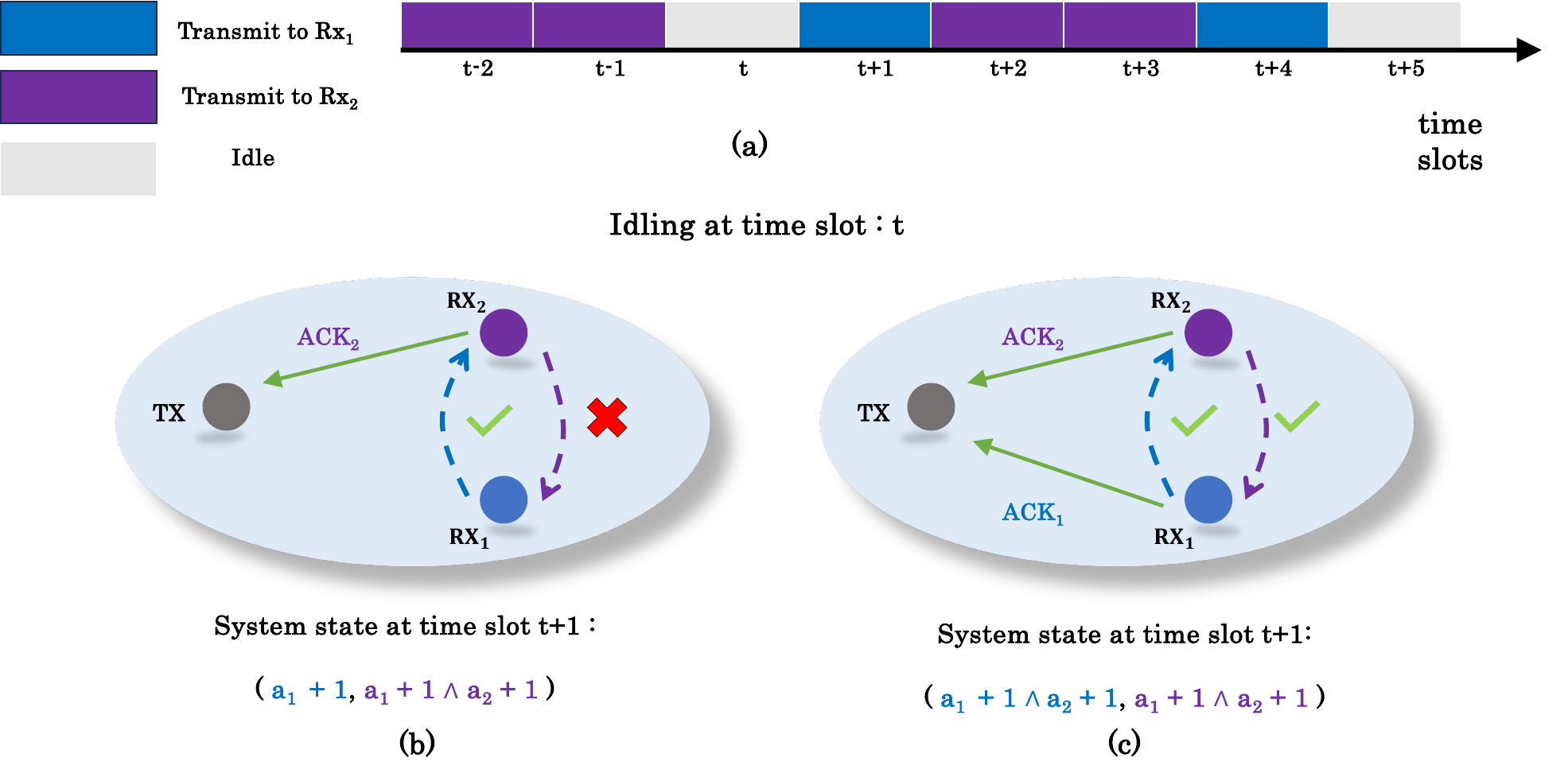} 
    \caption{(a) At each time slot, the transmitter has three action options: transmit to RX$_1$ (blue) or RX$_2$ (purple), or idle (gray). (b) If the TX idles in any time slot (e.g., in time slot $t$), both RX$_1$ and RX$_2$ can gossip. Whichever receiver has a successful update from the other lets the transmitter know by an ACK. (c) If both receivers have successful gossip updates, TX receives ACK's from both of them. TX maintains perfect RX age knowledge in all cases.}
    \label{fig:ACK}
\end{figure*}

Building upon these insights, we consider a setting in which there is a single source of information (a transmitter) and two gossiping receivers as shown in Fig. \ref{fig:GR System Model}. The transmitter may either update one receiver directly, after paying some cost, or remain idle to allow peer-to-peer gossiping. Both direct transmissions and gossiping are subject to failures with certain probabilities, e.g., one can be more reliable than the other. The goal is to devise a policy that minimizes the long-term average sum of AoI and transmission cost. We formulate this problem as an \emph{average-cost MDP} and solve it using the Relative Value Iteration (RVI) algorithm. When gossiping is relatively more reliable, it is shown that direct transmission to the older (higher age) receiver is optimal, provided that $(1)$ the minimum age of the two receivers is large enough, {\it and} $(2)$ the age difference between the two receivers is within a certain threshold. Otherwise, the optimal policy for the transmitter is to idle. Such thresholds depend on the system parameters, and in particular on the transmission cost. We describe this optimal policy as having an \emph{age-difference threshold structure with minimum age activation}. We prove several properties of the MDP's value function and formally prove the structure of the optimal policy. We compare the performance against multiple baselines used in the literature, including Max-Age-First (MAF), Max-Age-First with a threshold (MAFT), throughput-optimal (TPO) and random policies, and show that our optimal policy provides a relatively significant cost reduction in various settings.


\section{System Model and MDP Formulation}
We consider a time slotted system indexed by $t \in \{1, 2,...\}$. In each time slot, the transmitter (TX) can choose to transmit to (update) one of two receivers (RX), or it may idle. If the transmitter is idling, then both the receivers can gossip with each other. The TX always sends fresh updates, while RXs gossip with what they have. A direct transmission from TX to RX$_i$ is successful with probability $p_i$, and a gossip transmission from RX$_i$ to RX$_j$ is successful with probability $p_{vi}$, see Fig.~\ref{fig:GR System Model}. We assume the receivers send ACK/NACK feedback packets following each  transmitter's update attempt. Additionally, a receiver sends an ACK to the transmitter if it successfully receives a gossip update from the other receiver during an idling time slot. Therefore, the transmitter has perfect knowledge of both the receiver's ages due to these ACK/NACK packets. The process is illustrated in Fig. \ref{fig:ACK}.

Let $a_i(t)$ denote the AoI of RX$_i$. This is defined as $a_i(t)=t-\sigma_i(t)$, where $\sigma_i(t)$ denotes the time stamp of the latest successfully-received packet at RX$_i$. The system state is defined as the pair of AoI for the receivers, $s(t) = (a_1(t), a_2(t))$, and the state space is $\mathcal{S} = \{0, 1,2, \dots, A_{\max}\}^2$, where AoI saturates at a finite value $A_{\max}$. We assume a one-slot service delay: updates take one time slot to reach their destinations and be processed, and so the age becomes $1$ at RX$_i$ only after a fresh update has been transmitted to it from TX. The state $(0,0)$ is an initial state occurring only at $t=0$, and hence $(1,1)$ will appear at $t=1$, regardless of the TX's action at $t=1$, but is not reachable thereafter. The states $(1,0)$ and $(0,1)$ can not occur and are not part of the state space. 

Whenever RX$_i$ receives fresher information than what it has, either through direct transmission or gossiping, it updates its age and lowers it accordingly. Otherwise, its age evolves as  $a_i(t+1) = (a_i(t) + 1 ) \wedge A_{\max}$, where $a\wedge b$ denotes $\min(a,b)$. At time slot $t$, the transmitter can choose an action $u(t)$ from its action space $\mathcal{A}\in \{0,1, 2\}.$ Here, $u = 1$ or $u=2$ denotes transmission to RX$_1$ or RX$_2$, respectively, whereas $u=0$ denotes the transmitter idling. For notational convenience, we define the following shorthand notations:
\begin{align}
    m_i(t)&=(a_i(t) + 1 ) \wedge A_{\max},\quad i=1,2, \\
    m(t)&=m_1(t)\wedge m_2(t).
    \label{eq:notation}
\end{align}
Let us now observe the system dynamics, denoted by the conditional distribution $P(s(t+1)|s(t), u(t))$. For ease of presentation, we drop the dependency on $t$ and denote the next state as $s' = (a_1 (t+1), a_2(t+1))$. For $u = 1 $, we have
\begin{equation}
    P(s'|s, u = 1) = 
    \begin{cases}
    ( 1, m_2 ), & \text{w.p. }  p_1 \\[10pt]
    ( m_1, m_2 ), & \text{w.p. }  1- p_1\\
    \end{cases}.
    \label{eq:system dynamics u1}
\end{equation}
In a similar manner, we have 
\begin{equation}
    P(s'|s, u = 2) = 
    \begin{cases}
    ( m_1, 1 ), & \text{w.p. }  p_2 \\[10pt]
    ( m_1, m_2 ), & \text{w.p. }  1- p_2\\
    \end{cases}
    \label{eq:system dynamics u2}
\end{equation}
for $u=2$. If the transmitter idles ($u=0$), then the receivers can gossip with each other. We assume that successful gossiping from RX$_1$ to RX$_2$ is independent of that from RX$_2$ to RX$_1$. The system dynamics in this case are as follows:
\begin{equation}
    P(s'|s, u = 0) = 
    \begin{cases}
    ( m, m ), & \text{w.p. }  p_{v1} p_{v2} \\[6pt]
    ( m_1, m ), & \text{w.p. }  p_{v1}(1- p_{v2}) \\[6pt]
    ( m, m_2 ), & \text{w.p. }  (1- p_{v1})p_{v2} \\[6pt]
    ( m_1, m_2 ), & \text{w.p. }  (1- p_{v1})(1- p_{v2})  \\[6pt]
    \end{cases}.
    \label{eq:system dynamics u0}
\end{equation}

Next, we analyze the system cost, denoted by the function $c(s(t),u(t))$. First, whenever TX attempts a transmission, a non-negative fixed cost $C_{\mathrm{tx}}$ is incurred. Note that there is no cost incurred when idling, even if gossiping is successful. The transmitter incorporates the channel statistics knowledge (i.e., $p_1$, $p_2$, $p_{v1}$ and $p_{v2}$) and defines the system cost at a given time slot as the expected sum of both RXs' AoIs at the next time slot plus the transmission cost (if any). Thus, the cost $c(s, u)$ for each $u$ (after dropping the $t$) is given as follows. Following the system dynamics in \eqref{eq:system dynamics u1}, we have
\begin{equation}
    c(s, 1) = p_1(1 + m_2) + (1 - p_1)(m_1 + m_2) + C_{\mathrm{tx}}.
\end{equation}
To understand the above equation, observe that there are two possible outcomes when TX's action is $u=1$. Either the transmission is successful, w.p. $p_1$, in which case the AoI of RX$_1$ goes down to 1 (owing to one time slot delay), while the AoI of RX$_2$ increments to $m_2$. The state becomes $(1, m_2)$, and hence their sum becomes $1+m_2$. In the other case, w.p. $1- p_1$, direct transmission fails and the state becomes $(m_1, m_2)$. Additionally, we add the transmission cost $C_{\mathrm{tx}}$. Similarly, for $u=2$ one can derive from \eqref{eq:system dynamics u2} that
\begin{equation}
    c(s, 2) = p_2(m_1 + 1) + (1 - p_2)(m_1 + m_2) + C_{\mathrm{tx}}.
\end{equation}
When TX idles ($u=0$), there is no transmission cost $(C_{\mathrm{tx}} = 0)$. By the end of the time slot, the ages of RXs evolve to $m_1$ and $m_2$. The gossip model assumes independent bidirectional success between RXs, leading to four possible outcomes per idle slot. If gossip is successful bidirectionally, both ages become $m=\min(m_1, m_2)$. If gossip is successful uni-directionally, the sum of ages become $m_1+m$ or $m+m_2$ depending upon the direction of successful gossip. If both gossip links fail, the sum of ages becomes $m_1 + m_2$. Thus, for idling ($u=0$), we have from \eqref{eq:system dynamics u0} that
\begin{equation}
\begin{aligned}
    c(s, 0) ={}& (p_{v1} p_{v2})(2m) + p_{v1}(1-p_{v2})(m_1 + m) \\
               & + (1-p_{v1})p_{v2}(m + m_2) \\
               & + (1-p_{v1})(1-p_{v2})(m_1 + m_2).
\end{aligned}
\end{equation}

With these, we can now formulate our average-cost MDP. The objective is to find a policy $\pi : \mathcal{S} \to \mathcal{A}$ that minimizes the long-term average system cost, subject to the system dynamics, i.e., to minimize
\begin{equation}
\rho^{\pi}(s) = \limsup_{T \to \infty} \frac{1}{T} \mathbb{E}^{\pi} \left[ \sum_{t=0}^{T-1} c(s(t), \pi(s(t))) \Bigg| s_0=s \right],
\label{eq:average cost MDP}
\end{equation}
where $s_0$ denotes the initial state and $\mathbb{E}^{\pi}$ denotes expectation under policy $\pi$. We note that the optimal cost, $\rho^* = \min_\pi \rho^\pi(s)$, is independent of the initial state $s_0$ since the MDP can be shown to have a single recurrent class \cite{Puterman}. We use the Relative Value Iteration (RVI) algorithm to evaluate the optimal cost $\rho^*$ and the relative value function $U(s)$ by solving the Bellman optimality equation\cite{Bertsekas}:
\begin{align}
    \rho^* + U(s) &= \min_{u \in \mathcal{A}} \left\{ c(s,u) + \sum_{s' \in \mathcal{S}} P(s' \mid s,u) U(s') \right\} \notag \\
                  &= \min_{u \in \mathcal{A}} Q(s, u).
    \label{eq:bellman}
\end{align}
Finally, the optimal stationary policy $\pi^*(s)$ is derived by selecting the action that achieves the minimum in the Bellman equation for each state $s$:
\begin{equation}
\pi^*(s) = \arg \min_{u \in \mathcal{A}} Q(s, u).
\label{eq:bellan policy}
\end{equation}

\section{Optimal Policy Visualization and Performance}
\label{sec:numerical}

\begin{figure}[t]
    \centering
    \includegraphics[width=0.8\linewidth]{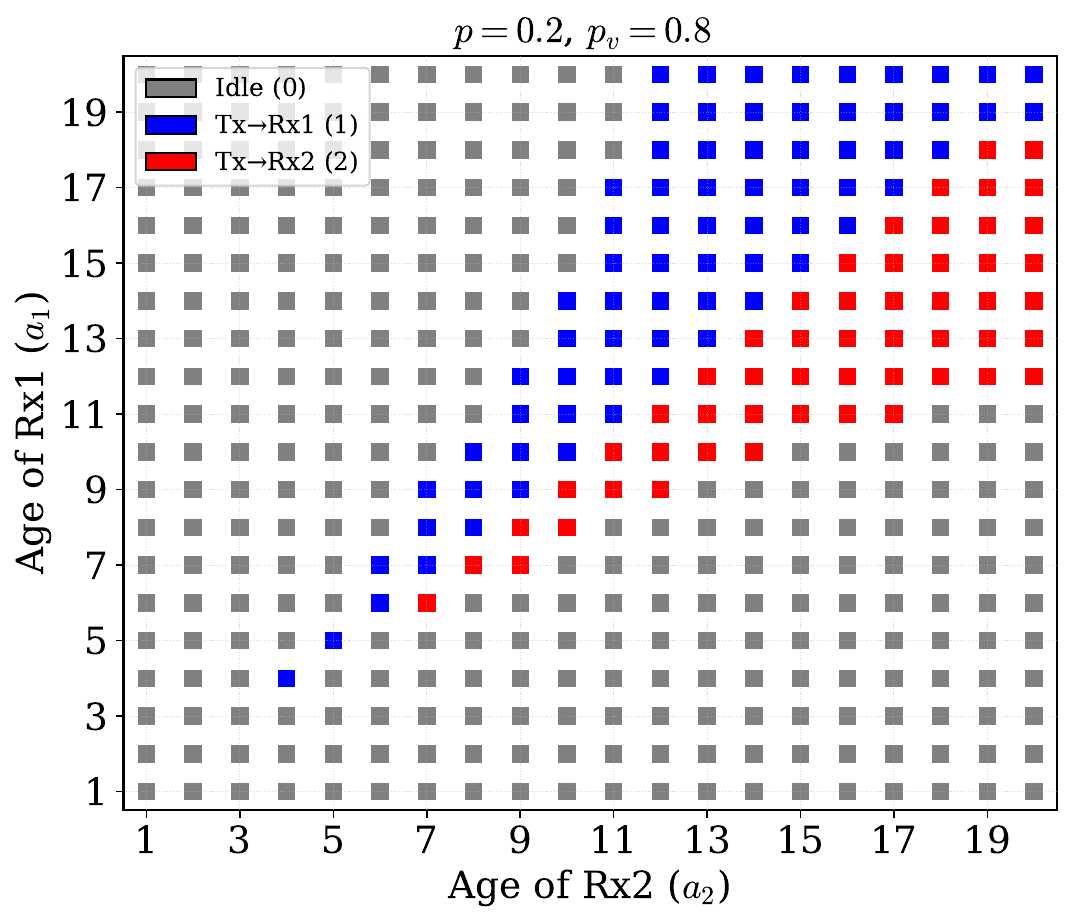} 
    \caption{Optimal policy visualization using RVI. Here, $C_{\mathrm{tx}}=1$.}
    \label{fig:Optimal Policy}
\end{figure}

In this section, we visualize the optimal policy and numerically compare the average cost for several cases of parameter values. Before doing so, let us summarize how we implement RVI to solve \eqref{eq:bellman}. We initialize with $U(s)^{(0)} = 0,~\forall s \in \mathcal{S}$, and for subsequent iterates $k>0$, we have
\begin{figure}[t]
    \centering
    \includegraphics[width=0.8\linewidth]{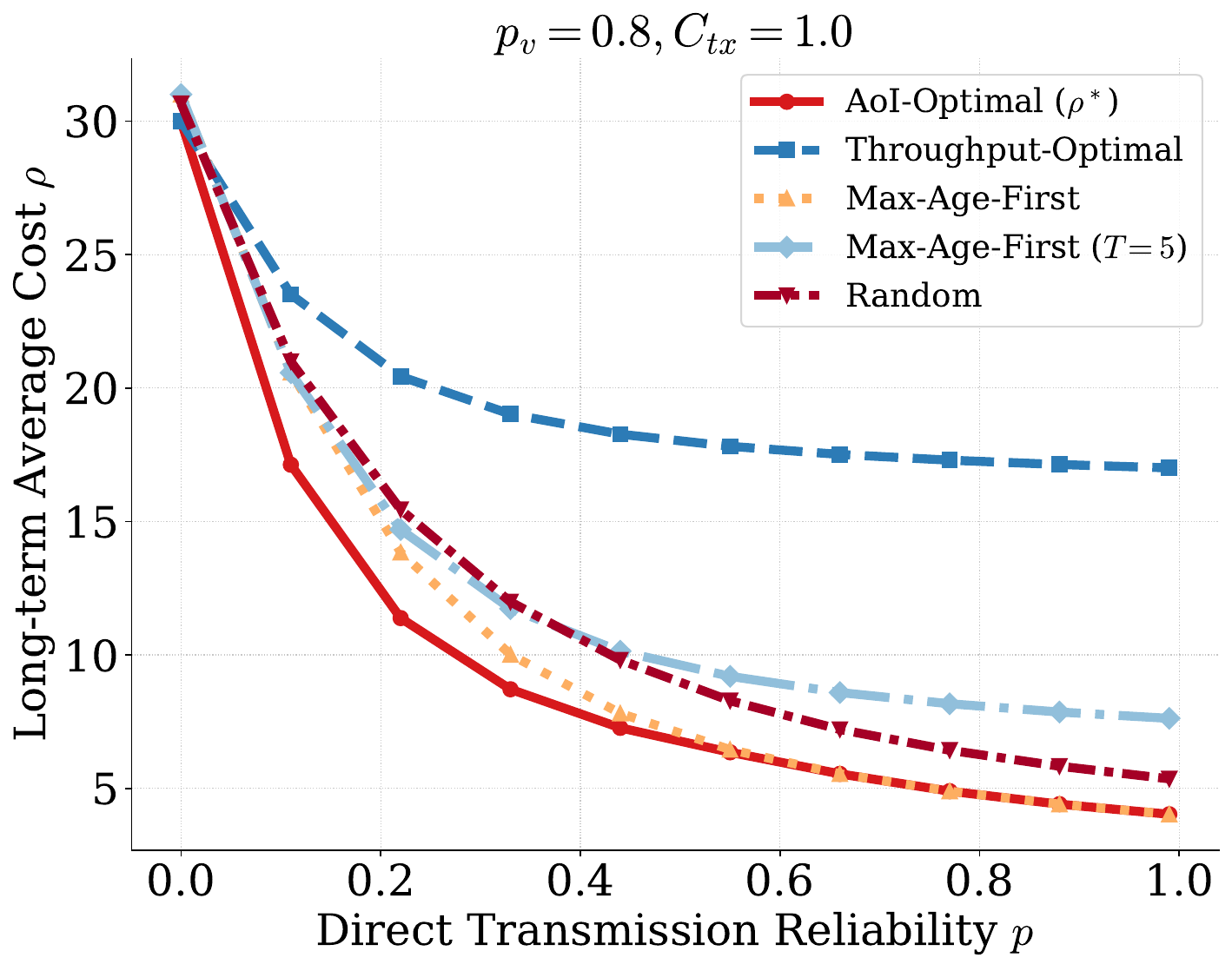} 
    \caption{Performance against varying direct transmission reliability.}.
    \label{fig:link reliability}
\end{figure} 
\begin{equation}
U^{(k+1)}(s) = \min_{u \in \mathcal{A}} \left\{ Q(s, u) \right\} - \rho_{k+1},
\label{eq:U^k+1}
\end{equation}
where $\rho_{k+1}$ is the relative value function for the reference state $s_{\text{ref}}$ (in our case $s_{\text{ref}} =  (1,1)$) in the $(k+1)$-th iteration. We subtract $\rho_{k+1}$ from each state and at each iteration the relative value function of the reference state becomes zero. Once $|| U^{k+1}(s) - U^{k}(s)|| < \epsilon $ holds for a chosen precision $\epsilon$, we stop iterating and retrieve $\rho^*, U$ and $\pi^*$.

Let us take a look at how the optimal policy looks like when gossiping is more reliable than direct transmission. We assume symmetric channels in which $p_v = p_{v1} = p_{v2}$ and $p_1 = p_2 = p$, with  $p_v > p  $. For the optimal policy in Fig.~\ref{fig:Optimal Policy}, we set $p_v = 0.8$, $p = 0.2$ and $C_{tx}=1$. We observe that the optimal policy in this case carries some structure, which we will formally analyze in the next section. First, we notice that the transmit decisions are symmetric with respect to the diagonal: the policy at state $(7,6)$ for instance is to transmit to RX$_1$, and if we switch the states to $(6,7)$ it is to transmit to RX$_2$. This applies to the whole state space at the optimal policy. Secondly, starting from the diagonal states, there exists a threshold along the horizontal direction after which there is a switch from transmitting to idling. Such threshold depends on the difference between $a_1$ and $a_2$. For example, from the state $(7,7)$ on the diagonal, as we keep increasing $a_2$, we go from transmitting to RX$_2$ to idling as we go from $(7,9)$ to $(7,10)$. After crossing the threshold, the optimal decision along the horizontal is to keep idle. This pattern repeats throughout the state space as we go horizontally or vertically from the diagonal (by symmetry). Finally, on the diagonal itself, as we go up on the diagonal the optimal decision for the transmitter switches from idling to transmitting, and it stays the same after the switch. For instance, there is a switch from idling to transmitting at state $(4,4)$.  After this switch, the decision along the diagonal is to always transmit. We note that transmitting to RX$_1$ or RX$_2$ is equally good on the diagonal (since $p_1 = p_2$); we arbitrarily choose $u=1$ (transmit to RX$_1$) in this case.

Next, we compare the optimal policy's performance against multiple baselines: MAF, serving the RX with the highest age; MAFT, which is the same as MAF but only activates after $\max_ia_i$ grows above a certain threshold $T$; TPO, which serves the best channel $\arg\max_ip_i$ (which is chosen arbitrarily here as RX$_1$ since $p_1=p_2$); and finally a policy that chooses randomly between all actions. All of the baseline policies (except the random one) are greedy in the sense that they always transmit (albeit after threshold for MAFT). We denote the optimal policy in the figures by AoI-optimal.

Fig.~\ref{fig:link reliability} shows how the direct transmission reliability $p$ affects the average cost. As $p$ increases, all policies achieve lower cost due to more successful transmissions, with AoI-optimal achieving the fastest decline as it balances transmission and idling to avoid redundant updates. We see that the performance gap between MAF and AoI-optimal is maximized in the region where gossiping is more reliable. This is further emphasized in Fig.~\ref{fig:gossip probability} where AoI-optimal effectively leverages gossiping to lower the overall cost, yielding a significant improvement. Meanwhile, greedy policies remain constant in Fig.~\ref{fig:gossip probability} since they do not adapt to gossiping reliability. Finally, Fig.~\ref{fig:transmission cost 2} illustrates the impact of transmission cost $C_{\mathrm{tx}}$. While the performance of some baselines are close to AoI-optimal for relatively low $C_{tx}$ (since transmission is inexpensive), the performance gap becomes significantly larger as $C_{\mathrm{tx}}$ increases, since almost all baselines are greedy.

\begin{figure}[t]
    \centering
    \includegraphics[width=0.8\linewidth]{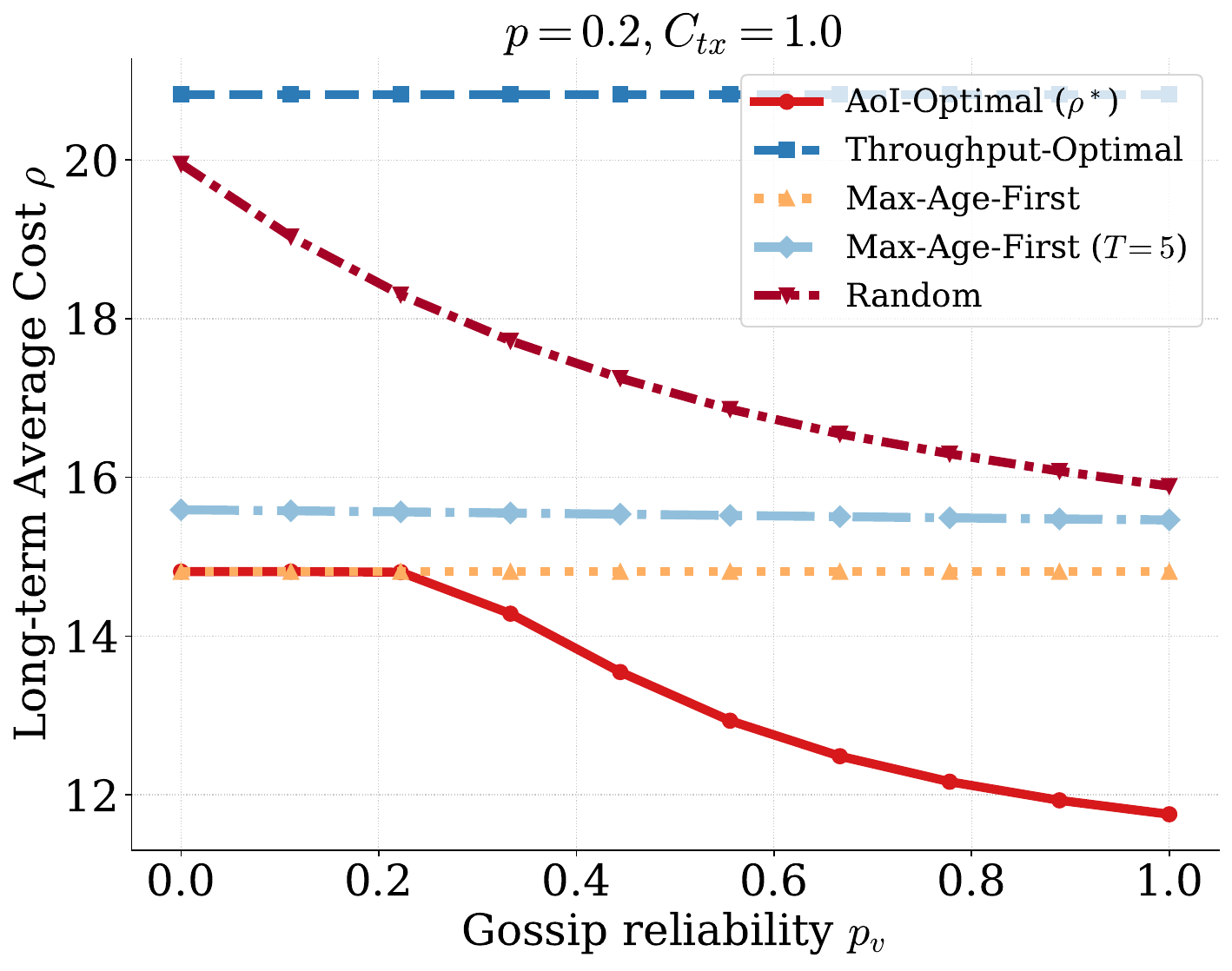} 
    \caption{Performance against varying gossiping reliability.}
    \label{fig:gossip probability}
\end{figure}

\section{Structural Analysis of the Optimal Policy}

In this section, we analyze the structural insights of the optimal RVI-based solution visualized in the previous section. We focus on the case of symmetric channels, with $p_v>p$ (gossiping is more reliable than direct transmission). Our goal is to provide reasoning behind the behavior of the optimal policy depicted in Fig.~\ref{fig:Optimal Policy}.

\subsection{Properties of the Relative Value Function}

We start by asserting some critical properties of the relative value function $U(s)$, stating its monotonicity and symmetry in the next two lemmas. Proofs are omitted due to space limits.

\begin{lemma}[Monotonicity]
\label{lemma:monotonicity}
The relative value function $U(s)$ is component-wise non-decreasing: for any states $(a_1, a_2)$ and $(a_1', a_2')$ in $\mathcal{S}$ such that $(a_1, a_2) \preceq (a_1', a_2')$, it holds that $U(a_1, a_2) \le U(a_1', a_2')$.
\end{lemma}

\begin{lemma}[Symmetry]
\label{lemma:symmetry}
The relative value function $U(s)$ is symmetric with respect to its arguments: $U(a_1, a_2) = U(a_2, a_1)$ for all $(a_1, a_2) \in \mathcal{S}$.
\end{lemma}

Without loss of generality, using Lemma \ref{lemma:symmetry}, we now restrict ourselves to the half plane where $a_1 < a_2$. For $a_1>a_2$, the policy switches analogously as shown in Fig.~\ref{fig:Optimal Policy}'s illustration in the previous section.

\subsection{Age-Difference Threshold Structure with Minimum Age Activation}

We now show that the optimal scheduling policy has a threshold structure based on the age difference between the two receivers and that it is only applied when the minimum age of the two receivers is large enough. We parameterize the state $s=(a_1, a_2)$ by its minimum age component, $m = \min(a_1, a_2)$, and the age difference $d = |a_1 - a_2|$. This means that for $a_1 < a_2$, the states have the parametrized form $s = (m, m+d)$ with $d \ge 0$. We now show that when the ages are unequal it is not optimal to serve the receiver with the lower age. The proof depends on Lemmas~\ref{lemma:monotonicity} and~\ref{lemma:symmetry}, and is omitted due to space limits.

\begin{proposition}
\label{prop:serve_older}
For any state $s=(a_1, a_2)$ with unequal ages ($a_1 \neq a_2$), transmitting to the receiver with the lower age is not optimal. That is, if $a_1 < a_2$, then $Q(s, 1) > Q(s, 2)$.
\end{proposition}

By Proposition \ref{prop:serve_older}, the decision for $d>0$ now reduces to comparing idling ($u=0$) with serving the higher age receiver (RX$_2$ in our decision region half plane). We define the advantage of transmitting over idling as $\Delta(s) \triangleq Q(s,2) - Q(s, 0)$. Thus, transmission is optimal if $\Delta(s) < 0$, and idling is optimal if $\Delta(s) > 0 $.

\begin{remark}
On the diagonal where $d=0$, actions $u=1$ and $u=2$ have the same $Q$-value ($Q((m,m), 1) = Q((m,m), 2)$). In case transmission is optimal, we arbitrarily choose $u=1$ in this case. This is why the diagonal line is blue-colored (transmit to RX$_1$) in Fig.~\ref{fig:Optimal Policy}.
\end{remark}

Since we work with the parameterized state $s = (m, m\, + d \,)$, we rewrite the dynamics in \eqref{eq:system dynamics u1} and \eqref{eq:system dynamics u2} as\begin{equation}
    P(s'|s, u = 1) = 
    \begin{cases}
    ( 1, m + d + 1 ), & \text{w.p. }  p \\[10pt]
    ( m + 1 , m + d + 1 ), & \text{w.p. }  1- p\\
    \end{cases},
\end{equation}
\begin{equation}
    P(s'|s, u = 2) = 
    \begin{cases}
    ( m+1 , \, 1 ), & \text{w.p. }  p \\[10pt]
    ( m + 1 , m + d + 1 ), & \text{w.p. }  1- p\\
    \end{cases}.
\end{equation}
Observe that idling and successful gossiping from an older (more stale) receiver to a fresher one both result in the same state space. Similarly, successful gossiping from a fresher receiver to an older one and successful gossiping in both directions result in the same state space. We group these cases and add their probabilities to get that, for idling, the state transition dynamics in \eqref{eq:system dynamics u0} become
\begin{equation}
    P(s'|s, u = 0) = 
    \begin{cases}
    ( m+1 , \, m+ 1  ), & \text{w.p. }  p_v \\[10pt]
    ( m + 1 , m + d + 1 ), & \text{w.p. }  1- p_v\\
    \end{cases}.
\end{equation}

\begin{figure}[t]
    \centering
    \includegraphics[width=0.8\linewidth]{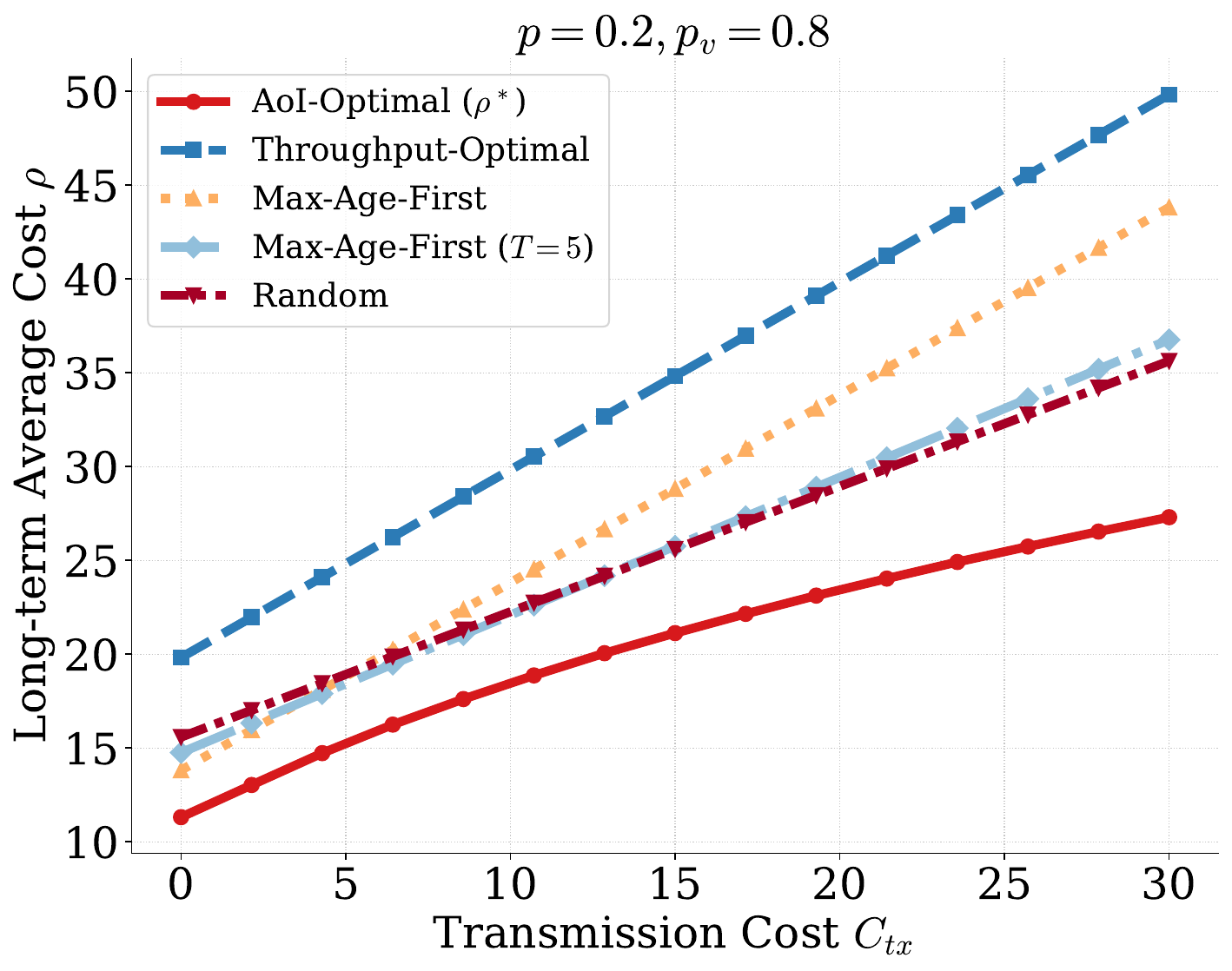} 
    \caption{Performance against varying transmission cost.}
    \label{fig:transmission cost 2}
\end{figure} 

We now proceed to prove that $\Delta(s)$ is monotonic in $d$. 
\begin{proposition}
\label{propostion:threshold}
For any fixed $m$,\, $\Delta(m, m+d)$ is strictly increasing in $d$. 
\end{proposition}

\begin{proof}
The Q-value for transmitting is given by
\begin{align}
\label{eq:q_value_transmit}
\begin{split}
Q(s, 2) ={}& p(m+2) + (1 - p)(2m+d+2) + C_{\mathrm{tx}} \\
& + p U(m+1, 1) + (1 - p) U(m+1, m+d+1),
\end{split}
\end{align}
and the Q-value for idling is given by
\begin{align}
\label{eq:q_value_idle}
\begin{split}
Q(s, 0 ) ={}& (1 - p_v)(2m+d+2) + p_v(2(m+1)) \\
& + (1 - p_v) U(m+1, m+d+1) \\
&+ p_v U(m+1, m+1).
\end{split}
\end{align}
Taking the difference $Q(s, 2) - Q(s, 0)$ and rearranging gives
\begin{align}
\begin{split}
\Delta(m, m+d) ={}& (p_v - p)d + (p_v - p)U(m+1, m+d+1) \\
&+C_{\mathrm{tx}} - mp + pU(m+1, 1)\\
&- p_vU(m+1, m+1)
\end{split}
\label{eq:delta cost}
\end{align}
We now examine the difference $\Delta(m, m+d+1) - \Delta(m, m+d)$. By the equation above, this difference can be simplified to
\begin{equation}
\label{eq:discrete_diff}
(p_v - p) \left[ 1 + U(m+1, m+d+2) - U(m+1, m+d+1) \right].
\end{equation}
Now for $p_v > p$, the term $(p_v-p)$ is strictly positive. By Lemma~\ref{lemma:monotonicity}, the difference of the $U$ terms is non-negative. Therefore, the entire expression in \eqref{eq:discrete_diff} is strictly positive, proving that $\Delta(m, m+d)$ is strictly increasing in $d$. \end{proof}

We now present the main structural result.

\begin{theorem}
\label{thm:switching}
Along the diagonal states \(s = (m,m)\), transmit is optimal if and only if
\begin{align}
C_{\mathrm{tx}} < p\big(m + U(m{+}1,m{+}1) - U(m{+}1,1)\big)
\label{eq:diag-cond}
\end{align}
holds. Further, for such $m$ there exists a threshold \(d_l(m)\) in the lateral directions (from the diagonal) after which the transmitter switches from \emph{transmit} to \emph{idle}. This is given by
\begin{align} \label{eq:lateral-threshold}
d_l(m) = \min\{\, d>0 : \Delta(m,m{+}d) > 0 \,\}.
\end{align}

Conversely, if the condition in \eqref{eq:diag-cond} does not hold, the optimal policy is to idle, starting from the diagonal state, remaining idle along the lateral directions without switching.
\end{theorem}

\begin{proof}
We start with the converse part. Recall that the optimal policy is to transmit when $\Delta(s) < 0$ and to idle when $\Delta(s) > 0$. By Proposition~\ref{propostion:threshold}, if $\Delta(m,m)\geq0$, then $\Delta(m,m+d)\geq0$ for all $d\geq0$, and hence idling is always optimal along the lateral lines if idling at the diagonal. Now let us evaluate \(\Delta(m,m)\) by setting \(d=0\) in \eqref{eq:delta cost} to get
\begin{align}
\Delta(m,m) &= (p_v-p)(0) + (p_v-p)U(m+1,m+1) \notag\\
&\quad + C_{\mathrm{tx}} - m p + pU(m+1,1) \notag\\
&\quad - p_vU(m+1,m+1) \notag\\
&=C_{\mathrm{tx}} - p\bigl(m + U(m+1,m+1) - U(m+1,1)\bigr) \label{eq:diag-delta}
\end{align}
Thus, \eqref{eq:diag-cond} does not hold when idling is optimal.

Now let us assume that \eqref{eq:diag-cond} holds and transmit is optimal. By Proposition~\ref{propostion:threshold}, there exists a threshold satisfying \eqref{eq:lateral-threshold} at which $\Delta(s)$ switches from negative to positive and hence idle becomes optimal.
\end{proof}

Theorem~\ref{thm:switching} shows that the optimal policy on the diagonal is to idle when the transmission cost $C_{\mathrm{tx}}$ is not justified by the age reduction a transmission may provide. This is true when the age is relatively low. However, once both ages grow high enough, the policy then intervenes by transmitting. Once activated along the diagonal, in the lateral directions when the age difference is beyond the thresholds, the policy idles again to let the receivers possibly update each other without incurring transmission cost. These set of results demonstrate the age-difference threshold structure of the optimal policy with minimum age activation.

The next corollary provides a sufficient condition for TX activation, inspired by the results of the theorem.

\begin{corollary}
    Let $m^*$ denote the minimum age of the two receivers after which the transmitter is activated. We have
    \begin{align}
        m^* \le \left\lceil \frac{C_{\mathrm{tx}}}{p} \right\rceil\triangleq \bar{m}.
    \end{align}
\end{corollary}

\begin{proof}
By Lemma~\ref{lemma:monotonicity}, the relative value function is component-wise non-decreasing, implying
\begin{align}
U(m{+}1,m{+}1) \ge U(m{+}1,1).
\label{eq:monotonic}
\end{align}
Substituting \eqref{eq:monotonic} into \eqref{eq:diag-delta}, we get the upper bound
\begin{align}
\Delta(m,m) \le C_{\mathrm{tx}} - p\,m.
\label{eq:diag-delta-bound}
\end{align}
Hence, $\Delta(m,m)<0,~\forall m>\bar{m}$. Thus, the activation minimum age $m^*$ cannot be larger than $\bar{m}$.
\end{proof}

Applying the above results on 
the optimal policy in Fig.~\ref{fig:Optimal Policy}, for $p\!=\!0.2$ and $C_{\mathrm{tx}}\!=\!1$, we get $\bar{m}=5$. That is, activation is guaranteed for $m\geq5$. Indeed, the policy begins transmitting from $m\!=\!4$.


\section{Conclusion}
This work studied freshness-aware scheduling in a single transmitter serving two gossiping receivers. We formulated the joint minimization of transmission cost and sum AoI as an average-cost MDP and established a sequence of structural properties for the optimal policy. In case of relatively reliable gossiping, analytical results established that the optimal policy exhibits an age-difference threshold structure with minimum age activation: once it is activated, the transmitter idles when receiver ages differ beyond a threshold, and before that it transmits to the higher age receiver. Numerical results showed that this policy significantly outperforms multiple baselines. Future work includes extensions to larger and more involved gossiping networks.

\bibliographystyle{unsrt} 
\bibliography{references}  
\end{document}